\documentclass[12pt]{amsart}
\usepackage{amssymb}
\usepackage{color}
\usepackage{amsmath,epic,curves,amscd}
\usepackage[english]{babel}
\usepackage{graphicx}
\usepackage{comment}
\usepackage{appendix}
\usepackage{mathdots}
\usepackage[all]{xy}
\usepackage{mathtools}
\usepackage{lscape}
\usepackage{stmaryrd}
\usepackage{mathabx}
\usepackage{multirow}
\newtheorem{claim}{}[section]
\newtheorem{theorem}[claim]{Theorem}

\newtheorem{lemma}[claim]{Lemma}
\newtheorem{proposition}[claim]{Proposition}
\newtheorem{corollary}[claim]{Corollary}

\theoremstyle{remark}

\renewenvironment{proof}{\noindent{\it Proof. \hskip0pt}}
                      {$\square$\par\medskip}
\allowdisplaybreaks \textwidth 15.5 true cm \textheight 24.2 true cm
\usepackage{color}

\begin{document}
\baselineskip 6.0 truemm
\parindent 1.5 true pc

\newcommand\lan{\langle}
\newcommand\ran{\rangle}
\newcommand\tr{\operatorname{Tr}}
\newcommand\ot{\otimes}
\newcommand\ttt{{\text{\bf t}}}
\newcommand\rank{\ {\text{\rm rank of}}\ }
\newcommand\choi{{\rm C}}
\newcommand\cp{{{\mathbb C}{\mathbb P}}}
\newcommand\ccp{{{\mathbb C}{\mathbb C}{\mathbb P}}}
\newcommand\pos{{\mathcal P}}
\newcommand\superpos{{{\mathbb S\mathbb P}}}
\newcommand\blockpos{{{\mathcal B\mathcal P}}}
\newcommand\jc{{\text{\rm JC}}}
\newcommand\dec{{\mathbb D}{\mathbb E}{\mathbb C}}
\newcommand\decmat{{\mathcal D}{\mathcal E}{\mathcal C}}
\newcommand\ppt{{\mathcal P}{\mathcal P}{\mathcal T}}
\newcommand\pptmap{{\mathbb P}{\mathbb P}{\mathbb T}}
\newcommand\xxxx{\bigskip\par {\color{red}\sf ========= Stop to read here ========= }}
\newcommand\join{\vee}
\newcommand\meet{\wedge}
\newcommand\ad{\operatorname{Ad}}
\newcommand{\id}{{\text{\rm id}}}
\newcommand\tsum{\textstyle\sum}
\newcommand\rk{{\text{\rm rank}}\,}
\newcommand\calI{{\mathcal I}}
\newcommand\sr{{\text{\rm SR}}\,}
\newcommand\conv{{\text{\rm conv}}\,}
\newcommand\pr{\prime}
\newcommand\e{\varepsilon}
\newcommand\tttt\intercal
\newcommand\sss{\bigstar}

\title{Choi matrices revisited. II}

\author{Kyung Hoon Han and Seung-Hyeok Kye}
\address{Kyung Hoon Han, Department of Data Science, The University of Suwon, Gyeonggi-do 445-743, Korea}
\email{kyunghoon.han at gmail.com}
\address{Seung-Hyeok Kye, Department of Mathematics and Institute of Mathematics,
    Seoul National University, Seoul 151-742, Korea}
\email{kye at snu.ac.kr}
\thanks{partially supported by NRF-2020R1A2C1A01004587, Korea}

\begin{abstract}
In this paper, we consider all possible variants of Choi matrices of linear maps,
and show that they are determined by non-degenerate bilinear forms on the domain
space. We will do this in the setting of finite dimensional vector spaces.
In case of matrix algebras, we characterize all variants of Choi matrices
which retain the usual correspondences between $k$-superpositivity and
Schmidt number $\le k$ as well as $k$-positivity and $k$-block-positivity.
We also compare de Pillis' definition [Pacific J. Math. {\bf 23}
(1967), 129--137] and Choi's definition [Linear Alg. Appl. \bf 10 \rm (1975), 285--290],
which arise from different bilinear forms.
\end{abstract}
\maketitle

\section{Introduction}

The celebrated Choi's theorem \cite{choi75-10} tells us that a linear map $\phi$ between matrices
is completely positive if and only if its Choi matrix
\begin{equation}\label{choi_def}
\choi_\phi=\textstyle\sum_{i,j}e_{ij}\ot\phi(e_{ij})=(\id\ot\phi)\left(\textstyle\sum_{i,j} e_{ij}\ot e_{ij}\right)
\end{equation}
is positive (semi-definite), where $\{e_{ij}\}$ is the system of the
usual matrix units. Around its appearance in 1970's,
completely positive maps have been generally accepted as right
morphisms between operator algebras reflecting noncommutative order
structures. Actually, important notions like nuclearity and
injectivity have been described in terms of completely positive
maps. See a survey article \cite{effros_survey}. The notion of Choi
matrices plays much more important roles in the current stage of
quantum information theory. See \cite{kye_lec_note}. A variant
$\sum_{i,j}e_{ji}\ot\phi(e_{ij})$ also has been used by de Pillis
\cite{dePillis} and Jamio\l kowski \cite{jam_72} prior to Choi's
theorem, to get the correspondences for Hermiticity preserving maps
and positive maps, respectively. The notion of Choi matrices can be
defined in various infinite dimensional cases
\cite{{holevo_2011},{holevo_2011_a},{li_du},{stormer_choi_mat},{Magajna_2021},{Friedland_2019},{Haapasalo_2020},{guddar}},
and for multi-linear maps
\cite{{kye_multi_dual},{han_kye_tri},{han_kye_multi}}.

It is clear that the definition (\ref{choi_def}) of the Choi matrix
depends on the choice of basis $\{e_{ij}\}$ of matrix algebra.
In this context, Paulsen and Shultz \cite{Paulsen_Shultz}
asked what happens when we replace matrix units in (\ref{choi_def}) by other bases of matrix algebras,
and gave a condition with which the Choi's correspondence is still valid.
On the other hand, the second author \cite{kye_Choi_matrix} considered a variant
\begin{equation}\label{choi-I}
\choi^\sigma_\phi
=(\id\ot\phi)(\choi_\sigma)
=\textstyle\sum_{i,j} e_{ij}\ot \phi(\sigma(e_{ij}))
\end{equation}
of Choi matrices by replacing $\textstyle\sum_{i,j} e_{ij}\ot
e_{ij}$ in (\ref{choi_def}) by an arbitrary matrix which is written
by $\choi_\sigma$ for a linear map $\sigma$ between the domain, and
showed that $\phi\leftrightarrow\choi^\sigma_\phi$ retains the
Choi's correspondence if and only if $\sigma$ is of the form $\ad_s$
defined by $\ad_s(x)=s^*xs$ for a nonsingular $s$ in the domain.

We go further in this paper to consider
all the possible expressions
\begin{equation}\label{choi_var}
\Gamma(\phi)=\textstyle\sum_i e_i\ot \phi(f_i)\in V\ot W
\end{equation}
which look like Choi matrices, for given linear maps $\phi:V\to W$
between vector spaces and pairs $(\{e_i\},\{f_i\})$ of bases of the
domain space $V$. We show that these expressions do not depend on
the choices of pairs of bases, but is determined by the
non-degenerate bilinear forms on the domain space $V$ defined by $\lan
e_i,f_j\ran=\delta_{ij}$, where $\delta_{ij}=1$ for $i=j$ and
$\delta_{ij}=0$ for $i\neq j$. Once we fix a basis of $V$, a
bilinear form on $V$ is determined by a linear isomorphism between
$V$, which is nothing but $\sigma$ appearing in (\ref{choi-I}) when
$V$ is a matrix algebra with a fixed basis $\{e_{ij}\}$ consisting of usual matrix units.
We also show that the
expression $\sum_i e_i\ot\phi(e_i)$ is possible with a single basis
$\{e_i\}$ if and only if the corresponding bilinear form is
symmetric.

We denote by $\mathbb P_k$ and $\superpos_k$ the convex cones of all $k$-positive and all $k$-superpositive
maps \cite{{cw-EB},{ssz}} between matrix algebras, respectively.
Note that $1$-superpositive maps are also called entanglement breaking maps
\cite{{ando-04},{hsrus}}. The correspondence $\phi\mapsto\choi_\phi$ through Choi matrices
between convex cones in the mapping space ${\mathcal L}(M_m,M_n)$
and those in $M_m\ot M_n$ can be summarized in the following diagram;
\begin{equation}\label{diagram}
\begin{matrix}
\phi\in{\mathcal L}(M_m,M_n): & \superpos_1  &\subset &\superpos_k
    &\subset &\cp &\subset &\mathbb P_k &\subset &\mathbb P_1\\
\\
\downarrow\phantom{XXXXXXX}  &\updownarrow  &&\updownarrow  &&\updownarrow  &&\updownarrow &&\updownarrow\\
\\
\choi_\phi\in M_m\ot M_n: & {\mathcal S}_1  &\subset &{\mathcal S}_k &\subset
    &\pos &\subset &\blockpos_k &\subset &\blockpos_1\\
\end{matrix}
\end{equation}
where ${\mathcal S}_k$ and $\blockpos_k$ denote the convex cones of all positive matrices with Schmidt number $\le k$
\cite{{eom-kye},{terhal-sghmidt}} and $k$-block-positive matrices, respectively. See
\cite{{kye_comp-ten},{kye_lec_note}} for surveys. Note that the correspondence $\cp\leftrightarrow \pos$ in the diagram
is just the Choi's theorem \cite{choi75-10}.
Recall that a bi-partite state is called separable \cite{Werner-1989} when it is belongs to ${\mathcal S}_1$,
and entangled if it is not separable.
We show that $\phi\mapsto \choi^\sigma_\phi$ retains the correspondence
$\mathbb P_k\leftrightarrow \blockpos_k$ in (\ref{diagram})
if and only if the correspondence $\superpos_k\leftrightarrow{\mathcal S}_k$ is retained if and only if
both $\sigma$ and $\sigma^{-1}$ are $k$-positive.

After we clarify the relations between non-degenerate bilinear pairings and the duality maps between
finite dimensional vector spaces in Section 2, we get the aforementioned results for vector spaces in Section 3.
In Section 4, we restrict ourselves to linear maps between matrix algebras, and show that $\sigma\mapsto\choi^\sigma_\phi$
retains the correspondences in (\ref{diagram}) if and only if both $\sigma$ and $\sigma^{-1}$ are $k$-positive.
Finally, we compare in Section 5 two definitions $\sum e_{ij}\ot\phi(e_{ij})$ and $\sum e_{ji}\ot\phi(e_{ij})$,
which correspond to the bilinear pairings $\tr(xy^\ttt)$ and $\tr(xy)$, respectively,
in connection with some identities involving the adjoints and tensor products of linear maps.

\section{Bilinear pairings and duality maps}

Suppose that $X$ and $Y$ are finite dimensional vector spaces over the complex field with a
bilinear pairing $\lan\ ,\ \ran_\pi$. Then we have the associated linear map
$D_\pi$ from $X$ to the dual space $Y^{\rm d}$ by
\begin{equation}\label{bilinear-linear}
D_\pi(x)(y)=\lan x, y\ran_\pi.
\end{equation}
Conversely, any linear map $D_\pi:X\to Y^{\rm d}$ gives rise to the bilinear pairing
by (\ref{bilinear-linear}).
The map $D_\pi$ is injective if and only if the associated bilinear pairing
satisfies the following condition:
\begin{enumerate}
\item[(${\rm N}_1$)]
for a given $x\in X$, we have $x=0$ if and only if $\lan
x,y\ran_\pi=0$ for every $y\in Y$.
\end{enumerate}
We consider the dual map $D_\pi^{\rm d}:Y\to X^{\rm d}$ of $D_\pi:X\to Y^{\rm d}$, which is given by
$$
D_\pi^{\rm d}(y)(x)=D_\pi(x)(y)=\lan x,y\ran_\pi.
$$
Then the dual map $D_\pi^{\rm d}$ is injective if and only if the following condition:
\begin{enumerate}
\item[(${\rm N}_2$)]
for a given $y\in Y$, we have $y=0$ if and only if $\lan x,y\ran_\pi=0$ for every $x\in X$.
\end{enumerate}
We consider one more condition on the spaces $X$ and $Y$:
\begin{enumerate}
\item[(${\rm N}_3$)]
$X$ and $Y$ share the same dimension.
\end{enumerate}
Then it is easy to check that two of (${\rm N}_1$), (${\rm N}_2$) and (${\rm N}_3$) imply the other.
When this is the case, we say that a bilinear pairing $\lan\ ,\ \ran_\pi$
is {\sl non-degenerate} and $D_\pi : X\to Y^{\rm d}$ is
its {\sl duality map}. We note that a linear map $D:X\to Y^{\rm d}$
is bijective if and only if it is the duality map of a non-degenerate bilinear pairing.

Suppose that $\{e_i:i\in I\}$ and $\{f_i:i\in I\}$ are bases of $X$ and $Y$ with the same dimension, respectively.
Then we may define the bilinear pairing on $X\times Y$ by
\begin{equation}\label{basis-bi}
\lan e_i,f_j\ran_\pi=\delta_{ij},
\end{equation}
that is, we define
$$
\left\lan \textstyle\sum_i\alpha_i e_i,\sum_i\beta_i f_i\right\ran_\pi
=\textstyle\sum_{i} \alpha_i \beta_i, \qquad \alpha_i, \beta_i \in \mathbb C.
$$
This bilinear pairing is automatically non-degenerate.

\begin{proposition}\label{basic_basis}
For a given bilinear pairing $\lan\ ,\ \ran_\pi$ between finite dimensional vector spaces
$X$ and $Y$ with the same dimension, the following are equivalent:
\begin{enumerate}
\item[{\rm (i)}]
$\lan\ ,\ \ran_\pi$ is non-degenerate;
\item[{\rm (ii)}]
there exist a basis $\{e_i:i\in I\}$ of $X$ and a basis $\{f_i:i\in I\}$ of $Y$  satisfying {\rm (\ref{basis-bi})};
\item[{\rm (iii)}]
for any given basis  $\{e_i:i\in I\}$ of $X$, there exists a unique
basis $\{f_i:i\in I\}$ of $Y$  satisfying {\rm (\ref{basis-bi})}.
\end{enumerate}
\end{proposition}

\begin{proof}
It remains to show the direction (i) $\Longrightarrow$ (iii).
For a given basis  $\{e_i:i\in I\}$ of $X$, we take the dual basis $\{e^\prime_i \in X^{\rm d} : i\in I\}$ given by
$e^\prime_i(e_j)=\delta_{ij}$. Because $D_\pi^{\rm d} : Y \to X^{\rm d}$ is bijective, we can find
$f_i\in Y$ satisfying $D_\pi^{\rm d}(f_i)=e^\prime_i$ for each $i\in I$. Then we have
$$
\lan e_i,f_j\ran_\pi=D_\pi^{\rm d}(f_j)(e_i)=e_j^\prime(e_i)=\delta_{i,j}
$$
for $i,j\in I$. The uniqueness follows from the non-degeneracy of bilinear form.
\end{proof}

A bilinear pairing on $X\times X$ is called a {\sl bilinear form} on $X$.
For a given basis $\{e_i:i\in I\}$ of a finite dimensional
vector space $X$, one can define the bilinear form by
\begin{equation}\label{basis-bi-form}
\lan e_i,e_j\ran_\pi=\delta_{ij}.
\end{equation}
Then this bilinear form is {\sl symmetric}, that is, $\lan x_1,x_2\ran_\pi=\lan x_2,x_1\ran_\pi$
for every $x_1,x_2\in X$. The bilinear form satisfies the relation (\ref{basis-bi-form})
if and only if its duality map $D_\pi$ satisfies
\begin{equation}\label{basis-dualitymap}
D_\pi(e_i)=e^\prime_i.
\end{equation}

\begin{proposition}\label{prop-basis-bil}
Suppose that $\lan\ ,\ \ran_\pi$ is a non-degenerate bilinear form on a finite dimensional
vector space $X$ over the complex field. Then the following are equivalent:
\begin{enumerate}
\item[{\rm (i)}]
$\lan\ ,\ \ran_\pi$ is symmetric;
\item[{\rm (ii)}]
there exists a basis $\{e_i:i\in I\}$ satisfying {\rm (\ref{basis-bi-form})};
\item[{\rm (iii)}]
there exists a basis $\{e_i:i\in I\}$ satisfying {\rm (\ref{basis-dualitymap})}.
\end{enumerate}
\end{proposition}

\begin{proof}
It suffices to prove the direction (i) $\Longrightarrow$ (ii).
Suppose that $\lan\ ,\ \ran_\pi$ is symmetric.
If $\lan x,x\ran_\pi=0$ for every $x\in X$ then
we have
$$
\begin{aligned}
0&=\langle x_1+x_2, x_1+x_2 \rangle_\pi\\
&= \langle x_1,x_1 \rangle_\pi + \langle x_1, x_2 \rangle_\pi
    + \langle x_2, x_1 \rangle_\pi + \langle x_2, x_2 \rangle_\pi\\
&= 2 \langle x_1, x_2 \rangle_\pi
\end{aligned}
$$
for every $x_1, x_2 \in X$, which is a contradiction. Therefore,
we can find $e_1$ such that $\langle e_1, e_1 \rangle_\pi=1$ by a scalar multiplication.
If $\{e_1\}$ is a basis then the proof is done.
In order to use the induction argument, we suppose that there exist linearly independent vectors $\{e_1,\dots, e_k\}$
satisfying {\rm (\ref{basis-bi-form})}, and we consider the subspace
$$
X_k=\{x\in X: \lan x,e_j\ran_\pi=0\ {\text{\rm for every}}\ j=1,2,\dots,k\}.
$$
For every $x\in X$, we have
$$
\left\lan x-\textstyle\sum_{i=1}^k\lan x,e_i\ran_\pi e_i, e_j\right\ran_\pi=0,\qquad j=1,2,\dots,k.
$$
and so $x-\sum_{i=1}^k\lan x,e_i\ran_\pi e_i$ belongs to $X_k$.
This means that the whole space $X$ is the direct sum of $X_k$ and the span of $\{e_1,\dots, e_k\}$.
Furthermore, the bilinear form $\lan\ ,\ \ran_\pi$ is still non-degenerate on $X_k$. Indeed,
if $x\in X_k$ and $\lan x,y\ran_\pi=0$ for every $y\in X_k$ then we see that $\lan x,y\ran =0$
for every $y\in X=X_k\oplus {\text{\rm span}\,}\{e_1,\dots e_k\}$, thus $x=0$.
Therefore, we can find $e_{k+1}\in X_k$ satisfying $\lan e_{k+1},e_{k+1} \ran_\pi = 1$
if $\{e_1,\dots, e_k\}$ is not a basis,
and this completes the induction argument.
\end{proof}

Given a basis $\mathcal B = \{b_i : i \in I\}$ of the matrix algebra
$M_m$, Paulsen and Shultz \cite{Paulsen_Shultz} defined the {\sl
duality map} by the linear map $D_{\mathcal B}:M_m\to M_m^{\rm d}$
sending $b_i$ to $b_i^\pr$. The duality map discussed
in this section generalizes this notion in the setting of vector
spaces. It is worthwhile to note that the duality map in
\cite{Paulsen_Shultz} does not depend on the choices of bases, but
on the bilinear forms determined by bases; if two bases determine
the same bilinear form, then the associated duality maps coincide.

Proposition \ref{basic_basis} (iii) tells us that a non-degenerate bilinear form
on $X$ is determined by a linear isomorphism $\sigma:X\to X$,
once we fix a  basis $\{\varepsilon_i\}$ of $X$,
which gives rise to the symmetric bilinear form by
$\lan\varepsilon_i,\varepsilon_j\ran=\delta_{ij}$.
In other words, we see that every possible non-degenerate bilinear pairing on $X$ is given by
\begin{equation}\label{symm}
\lan \varepsilon_i,\sigma(\varepsilon_j)\ran_\sigma=\delta_{ij},
\end{equation}
for a linear isomorphism $\sigma:X\to X$. It should be noted that
this correspondence between non-degenerate bilinear forms and linear
isomorphisms is valid when we fix a basis. In fact, when two bases
$\{e_i\}$ and $\{f_i\}$ induce different symmetric bilinear forms,
the identities $\lan e_i,\sigma(e_j)\ran=\delta_{ij}$ and $\lan
f_i,\sigma(f_j)\ran=\delta_{ij}$ give rise to different bilinear
forms. By the identity
$\lan\e_i,\sigma(\e_j)\ran_\sigma=\lan\e_i,\e_j\ran$, we have $\lan
x,\sigma(y)\ran_\sigma=\lan x,y\ran$, and so it follows that
\begin{equation}\label{symm-bas}
\lan x,y\ran_\sigma=\lan x,\sigma^{-1}(y)\ran,\qquad x,y\in X.
\end{equation}
For a given linear map $\sigma:X\to X$, we define the map $\sigma^\tttt:X\to X$ by
\begin{equation}\label{trans-bil}
\lan \sigma^\tttt(x),y\ran=\lan x,\sigma(y)\ran,\qquad x,y\in X.
\end{equation}
If we denote by $[\sigma]$ the matrix representing the map
$\sigma$ with respect to the fixed basis $\{\e_i\}$, then $[\sigma^\tttt]$ is nothing but the transpose of $[\sigma]$.

Instead of (\ref{symm}), one may get a symmetric bilinear form $\lan\ ,\ \ran_1$
by replacing $\e_i$ by $\sigma(\e_i)$ in both sides of
$\lan\e_i,\e_j\ran=\delta_{ij}$.
Then we have
$$
\lan x,y\ran_1=\lan\sigma^{-1}(x),\sigma^{-1}(y)\ran
=\lan x,(\sigma\circ\sigma^\tttt)^{-1}(y)\ran=\lan x,y\ran_{\sigma\circ\sigma^\tttt}.
$$
Compare with \cite[Definition 6]{Paulsen_Shultz}.

The bilinear form defined by (\ref{symm}), or equivalently by (\ref{symm-bas}),
is symmetric if and only if the matrix $[\sigma]$ is symmetric.
As an example, we consider a basis $\{e_i:i\in I\}$ of a finite-dimensional space and
a permutation $\sigma$ on the index set $I$, and define the bilinear form by
$\lan e_i,e_{\sigma(j)}\ran=\delta_{ij}$,
or equivalently $\lan e_i,e_j\ran=\delta_{i\sigma^{-1}(j)}$.
Then we see that this bilinear form is symmetric if and only if there exist
disjoint subfamilies of two-elements sets of $I$ such that $\sigma$ is the product of transpositions
on the subfamilies.
The permutation $\sigma$ induces the linear isomorphism $\sigma(e_i) = e_{\sigma(i)}$
by the abuse of notation. In this case, the matrix $[\sigma]$ is a symmetric matrix consisting of 0 and 1.
In the matrix algebras, we may fix the basis $\{e_{ij}\}$
and consider the identity permutation, to get the symmetric bilinear form
$$
\lan x,y\ran=\sum_{i,j=1}^n x_{ij}y_{ij}=\tr(xy^\ttt),\qquad x=[x_{ij}],\ y=[y_{ij}]\in M_n,
$$
where $x^\ttt$ denotes the usual transpose of the matrix $x$.
Another symmetric bilinear form
$$
\lan x,y\ran_\ttt=\sum_{i,j=1}^nx_{ij}y_{ji}=\tr(xy)=\lan x,y^\ttt\ran,\qquad x=[x_{ij}],\ y=[y_{ij}]\in M_n
$$
may be obtained by taking the permutation which sends $ij$ to $ji$.

\section{Isomorphisms between mapping spaces and tensor products}

For finite dimensional vector spaces $V$ and $W$, we are concerned
on the mapping space ${\mathcal L}(V,W)$ of all linear maps from $V$
into $W$, together with the tensor product $V\ot W$. In this
section, we will consider linear isomorphisms
\begin{equation}\label{gamma-map}
\Gamma: {\mathcal L}(V,W)\to V\ot W
\end{equation}
between these two spaces.
We will identify two spaces ${\mathcal L}(V,W)$ and $V^{\rm d}\ot W$.
In this identification, $v^\pr\ot w\in V^{\rm d}\ot W$ corresponds to the map $v \in V \mapsto v^\pr (v)w \in W$.

We proceed to show that a linear isomorphism $\Gamma$ in {\rm (\ref{gamma-map})}
has an expression (\ref{choi_var}) for a pair of bases if and only if it is the inverse of
the ampliation of the duality
map associated with the bilinear form given by the pair of bases.
$$
\xymatrix{\mathcal L(V,W) \ar[rrd]^\Gamma \ar[dd]_\simeq && \\
&& V \otimes W \ar[lld]^{D_\pi \otimes {\rm id}_W} \\
V^{\rm d} \otimes W && }
$$

\begin{lemma}\label{lemma}
Suppose that $\{e_i:i\in I\}$ and $\{f_i:i\in I\}$ are bases of a finite dimensional
vector space $V$, and $\lan\ ,\ \ran_\pi$ is a bilinear form on
$V$ with the duality map $D_\pi:V\to V^{\rm d}$. Then the following are equivalent:
\begin{enumerate}
\item[{\rm (i)}]
$\lan e_i,f_j\ran_\pi=\delta_{ij}$ for every $i,j\in I$;
\item[{\rm (ii)}]
for any vector space $W$, the identity $(D_\pi\ot \id_W)(\textstyle\sum_i e_i\ot\phi(f_i))=\phi$ holds
for every linear map $\phi : V \to W$.
\end{enumerate}
\end{lemma}

\begin{proof}
For the direction (i) $\Longrightarrow$ (ii), we note that the identity
$$
\begin{aligned}
(D_\pi\ot \id_W)(\textstyle\sum_i e_i\ot\phi(f_i))(f_j)
&=\textstyle\sum_i (D_\pi(e_i)\ot\phi(f_i))(f_j)\\
&=\textstyle\sum_i \lan e_i, f_j\ran_\pi \phi(f_i)\\
&=\phi(f_j),
\end{aligned}
$$
holds for every $\phi\in {\mathcal L}(V,W)$ and $j\in I$.

For the converse, we suppose that $(D_\pi \ot \id_W)(\textstyle\sum_i e_i\ot\phi(f_i))=\phi$ holds.
We take $W=V$, $\phi= {\rm id}_V$, and apply $f_j$ to the identity to get
$$
f_j={\rm id}_V(f_j)=\textstyle\sum_i \left(D_\pi(e_i)\ot f_i\right) (f_j)
=\textstyle\sum_i \lan e_i, f_j \ran_\pi f_i.
$$
Therefore, we have $\lan e_i, f_j \ran_\pi=\delta_{i,j}$.
\end{proof}

\begin{theorem}\label{vs_choi2}
Let $V$ and $W$ be finite dimensional vector spaces. For a given linear isomorphism $\Gamma$ in {\rm (\ref{gamma-map})},
the following are equivalent:
\begin{enumerate}
\item[(i)]
there exist bases $\{e_i\}$ and $\{f_i\}$ of $V$ such that $\Gamma(\phi)=\sum_i e_i\ot \phi(f_i)$
for every $\phi\in{\mathcal L}(V,W)$;
\item[(ii)]
for any basis $\{e_i\}$ of $V$, there exists a unique
basis $\{f_i\}$  of $V$ such that $\Gamma(\phi)=\sum_i e_i\ot \phi(f_i)$
for every $\phi\in{\mathcal L}(V,W)$;
\item[(iii)]
there exists a non-degenerate bilinear form $\lan\ ,\ \ran_\pi$
with the duality map $D_\pi:V\to V^{\rm d}$ such that
$\Gamma=D_\pi^{-1}\ot\id_W$.
\end{enumerate}
When these are the cases, a pair $(\{e_i\},\{f_i\})$  satisfies
{\rm (i)} if and only if $\lan e_i,f_j\ran_\pi=\delta_{i,j}$ holds.
\end{theorem}

\begin{proof}
Suppose that $\Gamma(\phi)=\sum_i e_i\ot \phi(f_i)$. Define
the bilinear form $\lan e_i,f_j\ran_\pi=\delta_{i,j}$.
By Lemma \ref{lemma}, we have
$$
((D_\pi\ot\id_W)\circ\Gamma)(\phi)
=(D_\pi\ot \id_W)(\textstyle\sum_i e_i\ot\phi(f_i))
=\phi.
$$
This proves the direction (i) $\Longrightarrow$ (iii).

For the direction (iii) $\Longrightarrow$ (ii), we
suppose that (iii) holds, and a basis  $\{e_i:i\in I\}$ is given. Then we can take a basis $\{f_i:i\in I\}$
satisfying $\lan e_i,f_j\ran_\pi=\delta_{i,j}$ by Proposition \ref{basic_basis}, and we have
$$
\Gamma(\phi)=\Gamma\circ(D_\pi\ot \id_W)(\textstyle\sum_i e_i\ot \phi(f_i))=\textstyle\sum_i e_i\ot \phi(f_i),
$$
by Lemma \ref{lemma}, as it was required.
For the uniqueness, we suppose that both bases $\{f_i\}$ and $\{\tilde f_i\}$
satisfy the condition. Then we have $\sum_i e_i\ot \phi(f_i)=\sum_i e_i\ot \phi(\tilde f_i)$. Taking
$e_i^\pr\ot\id_W$, we have $\phi(f_i)=\phi(\tilde f_i)$ for every $\phi\in{\mathcal L}(V,W)$. Therefore, we have
$f_i=\tilde f_i$ for each $i\in I$.

The direction (ii) $\Longrightarrow$ (i) is clear, and the last additional assertion follows from Lemma \ref{lemma} again.
\end{proof}

Theorem \ref{vs_choi2} tells us that the expression $\Gamma(\phi)=\sum_i e_i\ot \phi(f_i)$ depends on
the choice of a bilinear form $\lan\ ,\ \ran_\pi$ on $V$. Especially, this expression is invariant under the change
of bases, so far the pair of bases satisfies $\lan e_i,f_j\ran_\pi=\delta_{ij}$.
For an example, we consider the vector space $M_2$ of all $2\times 2$ complex matrices with the bilinear form
$$
\lan [x_{ij}], [y_{ij}]\ran=\sum_{i,j=1}^2 x_{ij}y_{ij}.
$$
Then the $2 \times 2$ Weyl basis consisting of
\begin{equation}\label{wepa}
E_1=\textstyle\frac 1{\sqrt 2}\left(\begin{matrix}1&0\\0&1\end{matrix}\right),\
E_2=\textstyle\frac 1{\sqrt 2}\left(\begin{matrix}1&0\\0&-1\end{matrix}\right),\
E_3=\textstyle\frac 1{\sqrt 2}\left(\begin{matrix}0&1\\1&0\end{matrix}\right),\
E_4=\textstyle\frac 1{\sqrt 2}\left(\begin{matrix}0&-1\\1&0\end{matrix}\right),
\end{equation}
considered in \cite[Remark 19]{Paulsen_Shultz}, satisfies the relation (\ref{basis-bi-form}). Because
the usual matrix units $\{e_{11}, e_{12}, e_{21}, e_{22}\}$ also satisfy the same relation,
we have
$$
\sum_{i=1}^4 E_i\ot \phi(E_i)=\sum_{i,j=1}^2e_{ij}\ot\phi(e_{ij}),
$$
which also can be checked directly.
We make it clear in the following theorem, whose proof is apparent by Lemma \ref{lemma}.

\begin{theorem}\label{corrs}
Suppose that $(\{e_i\},\{f_i\})$ and $(\{\tilde e_i\},\{\tilde f_i\})$ are pairs of bases
of a finite dimensional vector space $V$.
Then the following are equivalent:
\begin{enumerate}
\item[{\rm (i)}]
for any vector space $W$,
the identity $\sum_{i\in I}e_i\ot\phi(f_i)=\sum_{i\in I} \tilde e_i\ot \phi(\tilde f_i)$ holds
for every linear map $\phi:V\to W$;
\item[{\rm (ii)}]
two bilinear forms on $V$ defined by
$\lan e_i,f_j\ran = \delta_{ij}$ and $\lan \tilde e_i, \tilde f_j\ran=\delta_{ij}$ coincide.
\end{enumerate}
\end{theorem}

\begin{proof}
We denote by $\lan\ ,\ \ran_\pi$ and $\lan\ ,\ \ran_{\tilde\pi}$ the bilinear forms
given by $\lan e_i,f_j\ran_\pi = \delta_{ij}$ and $\lan \tilde e_i, \tilde f_j\ran_{\tilde\pi}=\delta_{ij}$,
respectively. Then the statement (i) holds if and only if $D_\pi=D_{\tilde\pi}$ by Lemma \ref{lemma}.
The proof is complete by the correspondence between the duality maps and bilinear forms.
\end{proof}

The appearance of $\delta_{ij}$ in Theorem \ref{corrs} (ii) is essential. Just existence
of a bilinear form satisfying $\lan e_i,f_j\ran = \lan \tilde e_i, \tilde f_j\ran$ does not guarantee
the identity in (i).
For an example, we take the vector space $\mathbb C^2$ with the bilinear form
$\lan x,y\ran=x_1y_1+x_2y_2$, which is non-degenerate. We take
$$
e_1=\left(\begin{matrix}1\\0\end{matrix}\right),\quad
e_2=\left(\begin{matrix}1\\1\end{matrix}\right),\quad
f_1=\left(\begin{matrix}1\\0\end{matrix}\right),\quad
f_2=\left(\begin{matrix}1\\-1\end{matrix}\right).
$$
Then we have $\lan e_i,e_j\ran=\lan f_i,f_j\ran$ for $i,j=1,2$. But,
$$
e_1\ot e_1+e_2\ot e_2=(2,1,1,1)^\ttt\neq (2,-1,-1,1)^\ttt=f_1\ot f_1+f_2\ot f_2,
$$
and so Theorem \ref{corrs} (i) does not hold for the identity map.

Now, we turn our attention to the question when the expression $\sum e_i\ot\phi(e_i)$ is possible
with a single basis $\{e_i\}$. The following is immediate by
Proposition \ref{prop-basis-bil} and Theorem \ref{vs_choi2}.

\begin{corollary}\label{coll}
Let $V$ and $W$ be finite dimensional vector spaces over the
complex field. For a given linear isomorphism $\Gamma$ in {\rm (\ref{gamma-map})},
the following are equivalent:
\begin{enumerate}
\item[(i)]
there exists a basis $\{e_i\}$ of $V$ such that $\Gamma(\phi)=\sum_i e_i\ot \phi(e_i)$;
\item[(ii)]
there exists a symmetric bilinear form $\lan\ ,\ \ran_\pi$ with the duality map $D_\pi:V\to V^d$ such that
$\Gamma=D_\pi^{-1}\ot\id_W$.
\end{enumerate}
When these are the cases, the basis $\{e_i\}$ satisfies {\rm (i)} if and only if $\lan e_i,e_j\ran_\pi=\delta_{i,j}$ holds.
\end{corollary}

Now, we consider the bilinear form on $M_2$ defined by
$$
\lan x,y\ran_\ttt=\tr(xy)=\sum_{i,j=1}^2 x_{ij} y_{ji},\qquad x=[x_{ij}], y=[y_{ij}]\in M_2,
$$
and the index set $I=\{11,12,21,22\}$. Then we have
$\lan e_{ij},e_{k\ell}\ran=1$ when and only when $(ij,k\ell)=(11,11), (12,21), (21,12), (22,22)$.
Therefore, we see that the pair
$$
(\{e_{ij}, ij\in I\},\{e_{ji}:ij\in I\})
$$ of bases determines the bilinear pairing $\lan\ ,\ \ran_\ttt$, and the corresponding Choi matrix is given by
$\sum_{i,j=1}^2e_{ij}\ot\phi(e_{ji})$ by Theorem \ref{vs_choi2}.
Corollary \ref{coll} tells us that there must exist a single basis determining this bilinear pairing.
Returning to matrices in (\ref{wepa}),
we have $\lan E_i,E_i\ran_\ttt=1$ for $i=1,2,3$ but $\lan E_4,E_4\ran_\ttt=-1$,
and so we have to replace $E_4$ by ${\rm i}E_4$ or $-{\rm i}E_4$
in order to get the relation (\ref{basis-bi-form}) with respect to this bilinear form.
Note that $E_1,E_2,E_3,{\rm i}E_4$ form Pauli matrices.
Applying Theorem \ref{vs_choi2} to Pauli matrices and the pair $(\{e_{ji}\}, \{e_{ij}\})$, we get
$$
\sum_{i=1}^3 E_i\ot \phi(E_i) + {\rm i}E_4 \ot \phi({\rm i}E_4)=\sum_{i,j=1}^2e_{ij}\ot\phi(e_{ji}).
$$
This example also has been considered in \cite{Paulsen_Shultz}
focusing on the choices of the bases, while we focus on the bilinear form produced by them.
See \cite{{Jiang_Luo_Fu_2013},{Frembs_Cavalcanti}} for discussions on the differences
between $\sum e_{ij}\ot\phi(e_{ij})$ and $\sum e_{ij}\ot\phi(e_{ji})$. See also Section 5.

\section{Choi matrices  in matrix algebras}

Now, we restrict our attention to the case of matrix algebras.
We take the usual matrix units $\{e_{ij}:i,j\in I\}$ as a fixed basis.
The corresponding symmetric bilinear form by Proposition \ref{prop-basis-bil} is given by
\begin{equation}\label{bi-lin-stan}
\lan x,y\ran=\sum_{i,j\in I}x_{ij}y_{ij}=\tr(xy^\ttt),\qquad x=[x_{ij}],\ y=[y_{ij}].
\end{equation}
When both $V=M_m$ and $W=M_n$ are matrix algebras, the correspondence $\phi\mapsto\Gamma(\phi)$ in Corollary \ref{coll}
given by the matrix units will be denoted by $\choi_\phi$. Then we have
\begin{equation}\label{stan-choi}
\begin{aligned}
\choi_\phi
&=\textstyle\sum_{i,j\in I} e_{ij}\ot \phi(e_{ij})\\
&=(\id_m\ot\phi)\left(\textstyle\sum_{i,j\in  I}e_{ij}\ot e_{ij}\right)\\
&=(\id_m\ot\phi)(\choi_{\id_m})\in M_m\ot M_n,
\end{aligned}
\end{equation}
for $\phi\in{\mathcal L}(M_m,M_n)$.
This is the Choi matrix of a linear map $\phi:M_m\to M_n$,
as it was defined in \cite{choi75-10}.

By Proposition \ref{basic_basis}, any non-degenerate bilinear form on $M_m$ is
determined by two bases $\{e_{ij}\}$ and $\{\sigma(e_{ij})\}$ for a linear isomorphism $\sigma:M_m\to M_m$.
Then, the bilinear form can be written by (\ref{symm-bas}).
The corresponding linear isomorphism $\Gamma$ given by Theorem \ref{vs_choi2} maps $\phi \in {\mathcal L}(M_m,M_n)$ to
the matrix
$$
\choi^\sigma_\phi:=\sum_{i,j\in I}e_{ij}\ot \phi(\sigma(e_{ij}))=\choi_{\phi\circ\sigma} \in M_m\ot M_n.
$$
Because $\choi^\sigma_\phi=(\id_m\ot\phi)(\choi_\sigma)$, we see that
$\choi^\sigma_\phi$ is obtained by replacing $\choi_{\id_m}$ by $\choi_\sigma$
in (\ref{stan-choi}),
as it  has been considered in \cite{kye_Choi_matrix}.
When $\sigma(x)=x^\ttt$, we have
$$
\lan x,y\ran_{\ttt}=\lan x,y^\ttt\ran=\tr(xy)=\sum_{i,j\in I}x_{ij}
y_{ji}
$$
for $x=[x_{ij}]$ and $y=[y_{ij}]$ in $M_m$, and we also have $\choi^\ttt_\phi=\sum_{i,j}e_{ij}\ot\phi(e_{ji})$,
which has been considered in \cite{dePillis,jam_72}.

For a linear map $\phi:M_m\to M_n$, the {\sl adjoint map} $\phi^*:M_n\to M_m$ is defined by
$$
\lan\phi(x),y\ran_{M_n}=\lan x,\phi^*(y)\ran_{M_m},\qquad x\in M_m,\ y\in M_n,
$$
using the bilinear forms $\lan\ ,\ \ran_{M_m}$ and $\lan\ ,\ \ran_{M_n}$ given by (\ref{bi-lin-stan}).
The relation with the dual map can be described by the commutative diagram:
$$
\xymatrix{M_n \ar[rr]^{\phi^*} \ar[d]^\simeq && M_m \ar[d]^\simeq \\
M_n^{\rm d} \ar[rr]^{\phi^{\rm d}} && M_m^{\rm d}}
$$
where column isomorphisms are the duality maps associated with the bilinear forms $\lan\ ,\ \ran_{M_m}$
and $\lan\ ,\ \ran_{M_n}$.
In fact, we have
$$
\begin{aligned}
(D_{M_m}(\phi^*(y)))(x)
&=\lan\phi^*(y),x\ran_{M_m}\\
&=\lan y,\phi(x)\ran_{M_n}\\
&=(D_{M_n}(y))(\phi(x))=(\phi^{\rm d}(D_{M_n}(y))(x),
\end{aligned}
$$
for every $x\in M_m$ and $y\in M_m$.
The Choi matrix $\choi_{\phi^*}\in M_n\ot M_m$ is nothing but the flip of $\choi_\phi\in M_m\ot M_n$.
See Section 5.

We also define the bilinear form on the mapping space ${\mathcal L}(M_m,M_n)$ by
$$
\lan\phi,\psi\ran=\lan\choi_\phi,\choi_\psi\ran,\qquad \phi,\psi\in {\mathcal L}(M_m,M_n),
$$
using the Choi matrices and the bilinear form on $M_m\ot M_n$ arising from matrix units.
See \cite{{sko-laa},{gks},{kye_comp-ten}}.
Then the identities
$$
\lan\phi,\psi^*\circ\sigma\ran
=\lan\psi\circ\phi,\sigma\ran
=\lan\psi,\sigma\circ\phi^*\ran
$$
are easily checked whenever the expression is legitimate.
For a subset $K\subset {\mathcal L}(M_m,M_n)$, we define the {\sl dual cone} $K^\circ$ by
$$
K^\circ=\{\phi\in{\mathcal L}(M_m,M_n): \lan \phi,\psi\ran\ge 0\ {\text{\rm  for every}}\ \psi\in K\}.
$$
It is well known that $K^{\circ\circ}$ is the smallest closed convex cone in ${\mathcal L}(M_m,M_n)$
containing the set $K$. For a convex cone $K$ in ${\mathcal L}(M_m,M_n)$, we use the notations
$$
K\circ\sigma = \{\phi\circ\sigma:\phi\in K\}
$$
and
$$
\choi^\sigma_K = \{ \choi^\sigma_\phi \in M_m\ot M_n : \phi \in K \}.
$$

We recall that the convex cones $\mathbb P_k$ and $\superpos_k$ are dual to each other, that is,
we have $\mathbb P_k^\circ=\superpos_k$.
The diagram (\ref{diagram}) tells us $\choi_{\mathbb P_k}=\blockpos_k$ and $\choi_{\superpos_k}={\mathcal S}_k$.
In the remainder of this section, we look for
conditions on a linear isomorphism $\sigma:M_m\to M_m$ under which the relations
$$
\choi^\sigma_{\mathbb P_k}=\blockpos_k
\qquad {\text{\rm and/or}}\qquad
\choi^\sigma_{\superpos_k}={\mathcal S}_k
$$
hold. For this purpose, it suffices to consider when
$\choi^\sigma_{\mathbb P_k}=\choi_{\mathbb P_k}$
and/or $\choi^\sigma_{\superpos_k}=\choi_{\superpos_k}$ hold.
We begin with the following:

\begin{proposition}
For a convex cone $K\subset {\mathcal L}(M_m,M_n)$ and a linear isomorphism
$\sigma:M_m\to M_m$, we have
$$
(K\circ\sigma^*)^\circ=K^\circ\circ\sigma^{-1}.
$$
\end{proposition}

\begin{proof}
We have
$$
\begin{aligned}
\phi\in (K\circ\sigma^*)^\circ
&\ \Longleftrightarrow\
 \lan\phi,\psi\circ\sigma^*\ran\ge 0\ {\text{\rm for every}}\ \psi\in K\\
&\ \Longleftrightarrow\
 \lan\phi\circ\sigma,\psi\ran\ge 0\ {\text{\rm  for every}}\ \psi\in K\\
&\ \Longleftrightarrow\
\phi\circ\sigma\in K^\circ\\
&\ \Longleftrightarrow\
 \phi\in K^\circ\circ\sigma^{-1},
\end{aligned}
$$
as it was required.
\end{proof}

Because $\choi^{\tau^*}_K=\choi_{K\circ\tau^*}$, we have the following equivalences;
\begin{equation}\label{choi_equiv}
\begin{aligned}
\choi_K=\choi^{\tau^*}_K
&\ \Longleftrightarrow\ K=K\circ\tau^*\\
&\ \Longleftrightarrow\ K^\circ=(K\circ\tau^*)^\circ\\
&\ \Longleftrightarrow\ K^\circ=K^\circ\circ\tau^{-1}\\
&\ \Longleftrightarrow\ \choi_{K^\circ}=\choi^{\tau^{-1}}_{K^\circ}.
\end{aligned}
\end{equation}

\begin{theorem}\label{choi-xxx}
For a given linear isomorphism $\sigma:M_m\to M_m$ and $k=1,2,\dots$, the following are equivalent:
\begin{enumerate}
\item[{\rm (i)}]
a linear map $\phi:M_m\to M_n$ is $k$-positive if and only if $\choi^\sigma_\phi \in M_m \otimes M_n$
is $k$-block-positive for every $n =1,2,\dots$;
\item[{\rm (ii)}]
a linear map $\phi:M_m\to M_n$ is $k$-superpositive if and only if $\choi^\sigma_\phi \in M_m \otimes M_n$
has Schmidt number $\le k$ for every $n =1,2,\dots$;
\item[{\rm (iii)}]
both $\sigma$ and $\sigma^{-1}$ is $k$-positive.
\end{enumerate}
\end{theorem}

\begin{proof}
By the one-to-one correspondence $\phi \leftrightarrow
\choi^\sigma_\phi$, the statements (i) and (ii) can be rephrased as
$\choi^\sigma_{\mathbb P_k}=\blockpos_k$ and
$\choi^\sigma_{\superpos_k}={\mathcal S}_k$, respectively.

We apply the equivalent statements in (\ref{choi_equiv}) with $K=\mathbb P_k$.
Then (iii) implies that $\mathbb P_k=\mathbb P_k\circ\sigma$, which implies
$\choi_{\mathbb P_k}=\choi^\sigma_{\mathbb P_k}$ by (\ref{choi_equiv}) with $\tau^*=\sigma$.
Thus we have $\choi^\sigma_{\mathbb P_k}=\blockpos_k$ and the statement (i).
We also have $\mathbb P_k=\mathbb P_k\circ(\sigma^{-1})^*$, which also implies the relation
$\choi^\sigma_{\mathbb P_k^\circ}=\choi_{\mathbb P_k^\circ}$
by (\ref{choi_equiv}) with $\tau=\sigma^{-1}$. Therefore, we have
$\choi^\sigma_{\superpos_k}=\choi^\sigma_{\mathbb P_k^\circ}=\choi_{\mathbb P_k^\circ}=
\choi_{\superpos_k}={\mathcal S}_k$,
which is the statement (ii).

For the direction (i) $\Longrightarrow$ (iii), we suppose that (i) holds with $m=n$.
Then we have $\mathbb P_k=\mathbb P_k\circ\sigma$ by (\ref{choi_equiv}) again.
Since $\id\in \mathbb P_k$, we see that both $\sigma$ and $\sigma^{-1}$ belong to $\mathbb P_k$.
The direction (ii) $\Longrightarrow$ (iii) can be seen from the relation
$\mathbb P_k=\mathbb P_k\circ (\sigma^{-1})^*$ by (\ref{choi_equiv}).
\end{proof}

It is well known \cite{schneider} that an order isomorphism between matrix algebras is of the form
$\ad_s$ or $\ad_s\circ\ttt$ for a nonsingular matrix $s$. See also \cite{{molnar},{Paulsen_Shultz},{semrl_souror}}.
Because the map $\ad_s\circ\ttt$ is never $k$-positive for $k\ge 2$
when $m,n\ge 2$, we see that an order isomorphism $\sigma$
is a $k$-positive with $k\ge 2$ if and only if it is a complete order isomorphism,
that is, both $\sigma$ and $\sigma^{-1}$ are completely positive if and only if
it is of the form $\ad_s$ with a nonsingular $s$.
This shows that \cite[Theorem 2.2]{kye_comp-ten} holds only when $s$ is nonsingular.
When $k=\min\{m,n\}$, we recover the following result \cite{{Paulsen_Shultz},{kye_Choi_matrix}}.

\begin{corollary}{\rm (\cite{{Paulsen_Shultz},{kye_Choi_matrix}})}\label{cp-pos-Choi}
For a given $\sigma:M_m\to M_m$, the following are equivalent:
\begin{enumerate}
\item[{\rm (i)}]
$\phi:M_m\to M_n$ is completely positive if and only if $\choi^\sigma_\phi$ is positive for all $n \in \mathbb N$;
\item[{\rm (ii)}]
$\sigma$ is a complete order isomorphism;
\item[{\rm (iii)}]
$\sigma=\ad_s$ for a nonsingular $s\in M_m$.
\end{enumerate}
\end{corollary}

We close this section to find conditions for $s\in M_m$ with which
the variant $\choi^{\ad_s}_\phi$ of Choi matrix can be expressed by $\sum b_{ij}\ot\phi(b_{ij})$
for a single basis $\{b_{ij}\}$ of $M_m$.

\begin{proposition}
Suppose that $s\in M_m$ is nonsingular with $\sigma=\ad_s$. Then the following are equivalent:
\begin{enumerate}
\item[{\rm (i)}]
$\choi^\sigma_\phi=\sum b_{ij}\ot\phi(b_{ij})$ for a basis $\{b_{ij}\}$ of $M_m$;
\item[{\rm (ii)}]
the bilinear form $\lan\ ,\ \ran_\sigma$ is symmetric;
\item[{\rm (iii)}]
$s$ is symmetric or anti-symmetric.
\end{enumerate}
\end{proposition}

\begin{proof}
It suffices to show the equivalence between (ii) and (iii). By
replacing $x$ and $y$ in $\langle x, \sigma^{-1}(y) \rangle =
\langle y, \sigma^{-1}(x) \rangle$ by $\sigma(x)$ and $\sigma(y)$
respectively, we see that (ii) is equivalent to
$$
\langle \sigma(x), y \rangle = \langle \sigma(y), x \rangle,\qquad x,y\in M_m.
$$
Since
$$
\lan \sigma(y),x\ran = \lan s^*ys,x\ran = \tr(s^*ys x^\ttt) = \tr(ys x^\ttt s^*) = \lan y, \bar sxs^\ttt\ran,
$$
this is equivalent to
the relation $s^*xs=\bar sxs^\ttt$, or equivalently
$(\bar s)^{-1}s^* x s s^{-\ttt} =  x$ for every $x\in M_m$, where we denote $s^{-\ttt}=(s^{-1})^\ttt=(s^\ttt)^{-1}$.
This is also equivalent to $\ad_{ss^{-\ttt}}=\id_{M_m}$.
Therefore, we conclude that the statement (ii) is equivalent to
$ss^{-\ttt}=zI_m$ for $|z|=1$, or equivalently
$$
s = z s^\ttt\quad {\text{\rm for a complex number}}\ z\ {\text{\rm with}}\  |z|=1.
$$
Since $s^\ttt = (z s^\ttt)^\ttt = z s$, we have $s = z^2 s$, thus
$z=\pm 1$. Therefore, we have $s=\pm s^\ttt$, and this completes the proof.
\end{proof}

Note that we actually have shown that the map $\ad_s:M_m\to M_m$ is {\sl symmetric} in the sense of
$(\ad_s)^\tttt=\ad_s$ if and only if the matrix $s\in M_m$ is symmetric or anti-symmetric.

\section{Choi matrices and bilinear forms}

Throughout this section, we fix a basis $\{\e_i\}$ of a finite dimensional vector space $V$.
Then we have seen that the corresponding bilinear form and Choi matrix are given by
$$
\lan \e_i,\e_j\ran_V=\delta_{ij}\qquad {\text{\rm and}}\qquad
\choi_\phi=\sum \e_i\ot \phi(\e_i),
$$
respectively, where $\phi:V\to W$ is a linear map.
If $W$ is also endowed with a bilinear form $\lan\ ,\ \ran_W$ associated with a fixed basis $\{\zeta_j\}$, then
we can define the {\sl adjoint} $\phi^*:W\to V$ by
$$
\lan \phi(x),y\ran_W=\lan x,\phi^*(y)\ran_V,\qquad x\in V,\ y\in W,
$$
as in the case of linear maps between matrix algebras.
For a given linear isomorphism $\sigma:V\to V$, we
have also defined the corresponding bilinear form and Choi matrix by
$$
\lan x,y\ran_\sigma=\lan x,\sigma^{-1}(y)\ran_V,\quad x,y\in V
\qquad {\text{\rm and}}\qquad
\choi^\sigma_\phi=\textstyle\sum_i \e_i\ot\phi(\sigma(\e_i)),
$$
for a linear map $\phi$ from $V$ to the other vector space $W$. If
a linear isomorphism $\tau:W\to W$ is given then we can define
the adjoint $\phi^{*_{\sigma,\tau}}:W\to V$ by
\begin{equation}\label{def-adj}
\lan \phi(x), y\ran_{\tau}
=\lan x, \phi^{*_{\sigma,\tau}}(y)\ran_{\sigma},\qquad x\in V,\ y\in W,
\end{equation}
which depends on the isomorphisms $\sigma$ and $\tau$.
The isomorphism $\sigma^\tttt$ defined in (\ref{trans-bil}) is nothing but the adjoint $\sigma^*$ defined above,
once we fix a basis. Nevertheless,
we retain the notation $\sigma^\tttt$ when $\sigma$ is an isomorphism from a vector space into itself,
which determines a bilinear form with respect to a fixed basis. We have the identity
$$
\langle x, y \rangle_{\sigma^\tttt}
= \langle x, (\sigma^\tttt)^{-1} (y) \rangle_V
= \langle \sigma^{-1} (x), y \rangle_V
= \langle y, \sigma^{-1} (x) \rangle_V
= \langle y, x \rangle_\sigma.
$$
We can also consider the bilinear form on the tensor product $V \otimes W$ by the fixed basis $\{\e_i \otimes \zeta_j\}$. Since
$\lan \e_i \otimes \zeta_j, \e_k \otimes \zeta_\ell \ran_{V \otimes W}
= \delta_{(i,j),(k,\ell)} = \lan \e_i, \e_k \ran_V \lan \zeta_j, \zeta_\ell \ran_W$,
we have
$$\lan x_1 \otimes y_1, x_2 \otimes y_2 \ran_{V \otimes W} = \lan x_1, x_2 \ran_V \lan y_1, y_2 \ran_W,
$$
which also implies that
$$
\lan x_1 \otimes y_1, x_2 \otimes y_2 \ran_{\sigma \otimes \tau}
= \lan x_1 \otimes y_1, \sigma^{-1}(x_2) \otimes \tau^{-1}(y_2) \ran_{V \otimes W}
= \lan x_1, x_2 \ran_\sigma \lan y_1, y_2 \ran_\tau.
$$

\begin{proposition}\label{ids}
Suppose that $\sigma:V\to V$ and $\tau:W\to W$ are linear isomorphisms.
For a linear map $\phi:V\to W$, we have the identities:
\begin{enumerate}
\item[{\rm (i)}]
$\lan\phi(x),y\ran_\tau=\lan\choi^\sigma_\phi, x\ot y\ran_{\sigma\ot\tau}$ for $x\in V$ and $y\in W$;
\item[{\rm (ii)}]
$\lan \choi^\tau_{\phi^{*_{\sigma,\tau}}},y\ot x\ran_{\tau\ot\sigma^\tttt}
=\lan \choi^\sigma_\phi, x\ot y\ran_{\sigma\ot \tau}$ for $x\in V$ and $y\in W$;
\item[{\rm (iii)}]
$\phi^{*_{\sigma,\tau}}=\sigma\circ\phi^*\circ\tau^{-1}$.
\end{enumerate}
\end{proposition}

\begin{proof}
We express arbitrary
$x\in V$ by $x=\sum_i x_i\sigma(\e_i)$, which implies
$$
\lan\e_j,x\ran_\sigma=\sum_i x_i\lan \e_j,\sigma(\e_i)\ran_\sigma=x_j.
$$
Then, we have
$x=\sum_i\lan\e_i,x\ran_\sigma\sigma(\e_i)$, and
$\phi(x)=\sum_i\lan\e_i,x\ran_\sigma\phi(\sigma(\e_i))$. Therefore, we have
\begin{equation}\label{id-choi-dual}
\begin{aligned}
\lan\phi(x),y\ran_\tau
&=\textstyle\sum_i\lan\e_i,x\ran_\sigma \lan\phi(\sigma(\e_i)),y\ran_\tau\\
&=\left\lan \textstyle\sum_i\e_i\ot \phi(\sigma(\e_i)), x\ot y\right\ran_{\sigma\ot\tau}\\
&=\lan\choi^\sigma_\phi, x\ot y\ran_{\sigma\ot\tau}.
\end{aligned}
\end{equation}
For the identity in (ii), we have
\begin{equation}\label{flip-choi}
\begin{aligned}
\lan \choi^\tau_{\phi^{*_{\sigma,\tau}}},y\ot x\ran_{\tau\ot\sigma^\tttt}
&=\lan \phi^{*_{\sigma,\tau}}(y),x\ran_{\sigma^\tttt}\\
&=\lan x, \phi^{*_{\sigma,\tau}}(y)\ran_\sigma\\
&=\lan \phi(x),y\ran_\tau
=\lan \choi^\sigma_\phi, x\ot y\ran_{\sigma\ot \tau}.
\end{aligned}
\end{equation}
Finally, we also have
$$
\begin{aligned}
\langle \phi(x), y \rangle_\tau
&= \langle \phi(x), \tau^{-1}(y) \rangle_W \\
&= \langle x, \phi^* \circ \tau^{-1}(y) \rangle_V
= \langle x, \sigma \circ \phi^* \circ \tau^{-1}(y) \rangle_\sigma,
\end{aligned}
$$
as it was required.
\end{proof}

In case of matrix algebras, we take the matrix units $\{e_{ij}\}$ as the fixed basis, as before.
From Proposition \ref{ids} (i), we have the identity
$$
\lan \choi_\phi, x\ot y\ran =\lan \phi(x),y\ran, \qquad \phi\in{\mathcal L}(M_m,M_n),\ x\in M_m,\ y\in M_n.
$$
This identity may be
considered as the definition of the Choi matrix, since it gives rise to the
entries of the Choi matrix when $x=e_{ij}$ and $y=e_{k\ell}$.
The identity
$$
\lan \choi_{\phi^*}, y\ot x\ran=\lan \choi_\phi, x\ot y\ran,  \qquad \phi\in{\mathcal L}(M_m,M_n),\ x\in M_m,\ y\in M_n
$$
from Proposition \ref{ids} (ii)
tells us that the Choi matrix $\choi_{\phi^*} \in M_n \otimes M_m$ of the adjoint $\phi^*$ is nothing but the flip of
the Choi matrix $\choi_\phi \in M_m \otimes M_n$.

Now, we use the bilinear forms given by $\lan x,y\ran_\ttt=\tr(xy)=\lan x, y^\ttt\ran$
with the corresponding Choi matrices
$$
\choi^\ttt_\phi=\choi_{\phi\circ\ttt}=\textstyle\sum e_{ij}\ot\phi(e_{ji})
=\textstyle\sum e_{ji}\ot\phi(e_{ij})=
(\ttt\ot\id)(\choi_\phi),
$$
as they have been used in \cite{{dePillis},{jam_72}}.
Then we have the identity
$$
\lan \choi^\ttt_{\phi}, x\ot y\ran_{\ttt\ot\ttt}
=\lan\phi(x),y\ran_\ttt,
$$
by Proposition \ref{ids} (i).
We write $\phi^\sss=\phi^{*_{\ttt,\ttt}}$ for simplicity, that is, we define $\phi^\sss$ by
$$
\lan \phi(x), y\ran_\ttt = \lan x, \phi^\sss(y) \ran_\ttt
$$
with respect to the bilinear forms $\lan\ ,\ \ran_\ttt$.
Then we also have
$$
\lan\choi^\ttt_{\phi^\sss},y\ot x\ran_{\ttt\ot\ttt}=\lan\choi^\ttt_\phi,x\ot y\ran_{\ttt\ot\ttt}
\qquad {\text{\rm and}}\qquad \phi^\sss= \ttt\circ\phi^*\circ\ttt,
$$
by Proposition \ref{ids} (ii) and (iii), respectively. Since
$$
\choi^\ttt_{\phi^\sss} = \choi_{\ttt \circ \phi^*} = [\phi^*(e_{ij})^\ttt] = ({\rm id} \otimes \ttt)(\choi_{\phi^*}),
$$
we see that $\choi_{\phi^\sss}^\ttt$ is the partial transpose of $\choi_{\phi^*}$ with respect to the second subsystem.
Since $\choi^\ttt_{\phi^\sss} = (\ttt \otimes {\rm id})(\choi_{\phi^\sss})$, we also have
$$
\choi_{\phi^\sss}=(\choi_{\phi^*})^\ttt.
$$
Here, $(\choi_{\phi^*})^\ttt$ is the global transpose of $\choi_{\phi^*}$ distinguished from
$\choi_{\phi^*}^\ttt$ in notations.
In particular, $\choi_{\phi^\sss}$ is the conjugation of the flip of $\choi_\phi$
when $\phi$ is Hermiticity preserving. This recovers \cite[Lemma 4.1.10]{stormer_book}.

For linear maps $\psi_i:M_{A_i}\to M_{B_i}$ for $i=1,2$ and $\phi:M_{A_1}\to M_{A_2}$, we have the
following simple identity
\begin{equation}\label{id}
(\psi_1\ot{\psi_2})(\choi_\phi)=\choi_{\psi_2\circ\phi\circ\psi_1^*}.
\end{equation}
See \cite{kye_comp-ten}. This identity is very useful to characterize mapping cones \cite{{stormer-dual},{gks}},
to provide unified simple arguments for various criteria arising from quantum information theory,
to get many equivalent statements to the PPT (Positive Partial Transpose) square conjecture by Christandl \cite{ppt}
which claims that the composition of two maps which are both completely positive and completely copositive is entanglement breaking.
When we endow $M_{A_i}$ and $M_{B_i}$ with bilinear forms arising from $\sigma_i$ and $\tau_i$, respectively,
it is reasonable to expect that $\psi_1\ot\psi_2$ sends $\choi^{\sigma_1}_\phi$ to $\choi^{\tau_1}_\Phi$
with $\Phi:=\psi_2\circ\phi\circ\psi_1^{*_{\sigma_1,\tau_1}}$. In fact, we have
$$
(\psi_1\ot \psi_2)(\choi^{\sigma_1}_\phi)
= (\psi_1\ot \psi_2)(\choi_{\phi \circ \sigma_1})
= \choi_{\psi_2 \circ \phi \circ \sigma_1 \circ \psi_1^*}
= \choi_{\Phi \circ \tau_1} \\
=\choi^{\tau_1}_\Phi.
$$
We will show this in the vector space level. See {\sc Figure 1}.

\begin{figure}\label{fig1}
\begin{center}
\setlength{\unitlength}{0.05 truecm}
\begin{picture}(140,75)
\put (-9,0){$(V_2,\sigma_2)$}
\put (17,3){\vector (1,0){80}}
\put (99,0){$(W_2,\tau_2)$}
\put (-9,50){$(V_1,\sigma_1)$}
\put (17,50){\vector (1,0){80}}
\put (97,53){\vector (-1,0){80}}
\put (99,50){$(W_1,\tau_1)$}
\put (5,45){\vector(0,-1){37}}
\put (-2,25){$\phi$}
\put (50,56){$\psi_1^{*_{\sigma_1,\tau_1}}$}
\put (50,43){$\psi_1$}
\put (50,6){$\psi_2$}
\put (110,45){\vector(0,-1){37}}
\put (113,25){$\Phi:=\psi_2\circ\phi\circ\psi_1^{*_{\sigma_1,\tau_1}}$}
\put (39,21){\line(1,0){36}}
\put (75,21){\line(0,1){12}}
\put (75,33){\line(-1,0){36}}
\put (39,33){\line(0,-1){12}}
\thicklines
\put (45, 25){$\psi_1\ot\psi_2$}
\put (17,26){\line(1,0){22}}
\put (17.3,25.9){\line(0,1){2}}
\put (75, 26){\vector(1,0){22}}
\put (17,28){\line(1,0){22}}
\put (75, 28){\vector(1,0){22}}
\end{picture}
\end{center}
\caption{The map $\psi_1\ot \psi_2$ sends $\choi^{\sigma_1}_\phi$ to $\choi^{\tau_1}_{\Phi}$.}
\end{figure}

\begin{proposition}
Suppose that $V_i$ and $W_i$ are endowed with the bilinear forms given by isomorphisms $\sigma_i$ and $\tau_i$,
respectively, and $\psi_i:V_i\to W_i$ is a linear map, for $i=1,2$.
For a given $\phi:V_1\to V_2$, put $\Phi=\psi_2\circ\phi\circ\psi_1^{*_{\sigma_1,\tau_1}}$.
Then we have
$$
(\psi_1\ot \psi_2)(\choi^{\sigma_1}_\phi)=\choi^{\tau_1}_\Phi.
$$
\end{proposition}

\begin{proof}
We have $(\psi_1\ot \psi_2)(\choi^{\sigma_1}_\phi)
=\textstyle\sum_i \psi_1(\e_i)\ot (\psi_2\circ\phi\circ\sigma_1)(\e_i)$
with the fixed basis $\{\e_i\}$. Therefore, we have
$$
\begin{aligned}
\lan (\psi_1\ot \psi_2)(\choi^{\sigma_1}_\phi), y_1\ot y_2\ran_{\tau_1\ot \tau_2}
&=\textstyle\sum_i \lan \psi_1(\e_i),y_1\ran_{\tau_1}
    \lan(\psi_2\circ\phi\circ\sigma_1)(\e_i), y_2\ran_{\tau_2}\\
&=\textstyle\sum_i \lan \e_i,\psi_1^{*_{\sigma_1,\tau_1}}(y_1)\ran_{\sigma_1}
    \lan\e_i, (\psi_2\circ\phi\circ\sigma_1)^{*_{\sigma_1,\tau_2}}(y_2)\ran_{\sigma_1}.
\end{aligned}
$$
Now, we use the identity $x=\sum_i\lan\e_i,x\ran_{V_1}\e_i$, to proceed
$$
\begin{aligned}
&=\textstyle\sum_i \lan \e_i,(\sigma_1^{-1}\circ\psi_1^{*_{\sigma_1,\tau_1}})(y_1)\ran_{V_1}
    \lan\e_i, (\psi_2\circ\phi\circ\sigma_1)^{*_{\sigma_1,\tau_2}})(y_2)\ran_{\sigma_1}\\
&=\lan (\sigma_1^{-1}\circ\psi_1^{*_{\sigma_1,\tau_1}})(y_1),
     (\psi_2\circ\phi\circ\sigma_1)^{*_{\sigma_1,\tau_2}}(y_2)\ran_{\sigma_1}\\
&=\lan (\psi_2\circ\phi\circ\sigma_1)\circ(\sigma_1^{-1}\circ\psi_1^{*_{\sigma_1,\tau_1}})(y_1),y_2\ran_{\tau_2}\\
&=\lan \Phi(y_1),y_2\ran_{\tau_2},
\end{aligned}
$$
from which we get the required identity by Proposition \ref{ids} (i).
\end{proof}

If we endow matrix algebras with bilinear forms $\lan\ ,\ \ran_\ttt$, then we have
$$
(\psi_1\ot \psi_2)(\choi^{\ttt}_\phi)=\choi^{\ttt}_{\psi_2\circ\phi\circ\psi_1^\sss},
$$
which is basically same as (\ref{id}).
If we replace $\phi$ by $\phi \circ \ttt$, then we get the relation
$$
(\psi_1\ot\psi_2)(\choi_\phi)=\choi_{\psi_2\circ\phi\circ\ttt\circ\psi_1^\sss}^\ttt.
$$
Discussions in the section are summarized in {\sc Table 1}.

\begin{table}\label{table}
\begin{center}
\renewcommand{\arraystretch}{1.4}
\begin{tabular}{|c|rl|}
\hline
bilinear form &$\lan x,y\ran=\tr(xy^\ttt)=\sum x_{ij}y_{ij}$
   &\vline\quad  $\lan x,y\ran_\ttt=\tr(xy)=\sum x_{ij}y_{ji}$\\
\hline
\multirow{3}{6em}{Choi matrix} &$\choi_\phi=\sum e_{ij}\ot\phi(e_{ij})$ &\vline\quad
       $\choi^\ttt_\phi=\sum e_{ij}\ot\phi(e_{ji})$\\
& $\lan\choi_\phi,x\ot y\ran=\lan\phi(x),y\ran$
      &\vline\quad $\lan\choi^\ttt_\phi,x\ot y\ran_{\ttt\ot\ttt}=\lan\phi(x),y\ran_\ttt$\\ \cline{2-3}
& $(\ttt\ot \id)(\choi_\phi)$ &$=\quad \choi^\ttt_\phi$\\
\hline
\multirow{2}{4em}{adjoint} &$\lan\phi(x),y\ran = \lan x,\phi^*(y)\ran$ &\vline\quad
   $\lan \phi(x),y\ran_\ttt = \lan x,\phi^\sss(y)\ran_\ttt$ \\ \cline{2-3}
& $\ttt\circ\phi^*\circ\ttt$ &$=\quad \phi^\sss$\\
\hline
\multirow{5}{6em}{Choi matrix of adjoint} & $\lan\choi_{\phi^*},y\ot x\ran=\lan\choi_\phi,x\ot y\ran$ &\vline\quad
    $\lan\choi^\ttt_{\phi^\sss},y\ot x\ran_{\ttt\ot\ttt}=\lan\choi^\ttt_\phi,x\ot y\ran_{\ttt\ot\ttt}$\\
 &$\choi_{\phi^*}=$ flip of $\choi_{\phi}$
     &\vline\quad $\choi^\ttt_{\phi^\sss}=$ flip of $\choi^\ttt_{\phi}$\\ \cline{2-3}
  & {$\lan \choi_{\phi^*}, x^\ttt \otimes y \ran$} &$=\lan\choi^\ttt_{\phi^\sss},x\ot y\ran_{\ttt\ot\ttt}$\\
& {$({\rm id} \otimes \ttt)(\choi_{\phi^*})$} &{$=\quad \choi_{\phi^\sss}^\ttt$}\\
& $(\choi_{\phi^*})^\ttt$        &$=\quad \choi_{\phi^\sss}$\\
\hline
\multirow{2}{7em}{tensor products of maps} & $(\psi_1\ot\psi_2)(\choi_\phi)=\choi_{\psi_2\circ\phi\circ\psi_1^*}$
              &  \vline\quad $(\psi_1\ot \psi_2)(\choi^{\ttt}_\phi)
              =\choi^{\ttt}_{\psi_2\circ\phi\circ\psi_1^\sss}$\\ \cline{2-3}
& $(\psi_1\ot\psi_2)(\choi_\phi)$ &
$=\quad \choi_{\psi_2\circ\phi\circ\ttt\circ\psi_1^\sss}^\ttt$ \\
\hline
\end{tabular}
\caption{Comparison between two bilinear forms $\tr(xy^\ttt)$ and $\tr(xy)$
 on matrix algebras together with associated Choi matrices and
adjoints.}
\end{center}
\end{table}

\end{document}